\documentclass[showkeys,aip,graphicx,amsmath,superscriptaddress,amssymb]{revtex4-2}

\usepackage{lineno,hyperref}
\usepackage[normalem]{ulem}
\usepackage{graphicx}
\usepackage{amsmath,bm,amssymb,amsfonts}
\usepackage{color}
\usepackage{booktabs}
\usepackage{array}
\usepackage{braket}
\usepackage{mathtools}
\usepackage{amsthm}
\usepackage[capitalize]{cleveref}

\theoremstyle{plain}
\newtheorem{result}{Stage}

\theoremstyle{definition}

\newtheorem*{advantage}{Quantum statistical-memory advantage}

\newcommand{\nq}{n_q}

\begin{document}

\title{Quantum Statistical Memory Advantage Reveals Predictive Structure in Chaotic Invariant Measures}

\author{Maida Wang}
\affiliation{Centre for Computational Science, University College London, London, UK}

\author{Xiao Xue}
\affiliation{Centre for Computational Science, University College London, London, UK}

\author{Minh Chung}
\affiliation{Leibniz Supercomputing Centre of the Bavarian Academy of Sciences and Humanities, Garching, Germany}

\author{Peter V. Coveney}
\email{p.v.coveney@ucl.ac.uk}
\affiliation{Centre for Computational Science, University College London, London, UK}
\affiliation{Centre for Advanced Research Computing, University College London, London, UK}


\begin{abstract}
Long-horizon prediction of a chaotic system is governed by fidelity to its invariant measure, and which parts of that measure a learned predictor needs is a physical question not known in advance. We establish a quantum statistical memory advantage for interrogating it. A quantum statistical prior (Q-Prior), trained once on classical data, stores an efficiently preparable $k$-point marginal in polynomially many circuit parameters against exponentially many for explicit tabulation. Collective Bell measurements on two copies then estimate any \emph{post hoc} Pauli-expectation magnitude at a copy cost independent of register size, with a worst-case exponential separation from single-copy protocols. A single Bell dataset supports an entire family of candidate observables at a cost growing only logarithmically in the family size, so an exponentially large candidate space stays open for later analysis. We implement the protocol on IQM superconducting processors with two-copy registers of up to 54 physical-qubit chips. We apply this to turbulent channel flow and ERA5 forecasting, resolving invariant structure by statistical sector and order. Turbulent phase correlations persist through eighth order, while most predictive gains arise from low-order constraints; in ERA5, planetary-wave phase coherence complements covariance regularisation, identifying distinct low-order sectors with predictive value. Quantum statistical memory is therefore both a near-term computational resource and an instrument for identifying which invariant structures matter for prediction.
\end{abstract}

\maketitle


\section{Introduction}
Data-driven machine learning predictors of chaotic flows fail in a characteristic way~\cite{price2025probabilistic,lam2023learning,schiff2024dyslim}. Pointwise accuracy decays on the scale of the Lyapunov time, which is unavoidable, but the rollout also drifts off the attractor, which is not, and the drift is not visible in pointwise error scores. The object that detects it is the invariant measure, a canonical object of ergodic theory~\cite{eckmann1985,petersen1989ergodic,budivsic2012applied} rather than an arbitrary generative target~\cite{bengio2021gflownet,stokes2020antibiotic}, and the long-horizon stability of a learned predictor is governed by its statistical fidelity to that measure~\cite{schiff2024dyslim}.

An invariant measure is not a single statistic. It carries a hierarchy of candidate structures: marginals at different spatial orders, long-range correlations, phase relationships between distant regions, multi-site closures above pairwise order, and statistical organisation that can persist well beyond second order. Which of these a given predictor actually needs is not known in advance, and it is not obviously the lowest-order one. The question this work is built around is therefore which parts of an invariant measure matter for stabilising learned long-horizon prediction.

A representation that commits to one statistic when it is built can only answer the question it was built for. Interrogating the hierarchy instead calls for a representation that is compact enough to hold a high-order object and that can still be queried for structures selected after it has been built and measured. Compactness and the ability to answer queries selected after measurement are separate requirements, and it is their conjunction that a quantum register supplies.

On the other side, roadmaps to fault-tolerant quantum computing remain uncertain and the transition will be gradual~\cite{Preskill_2018,katabarwa2024early,googlequantum2025below,kim2023evidence}, many reported pre-fault-tolerant speedups have since been matched by better classical algorithms~\cite{huang2026vast}, and the clearest demonstrations of quantum advantage establish separations on distributions constructed to be classically hard rather than on scientific problems~\cite{arute2019supremacy,zhong2020jiuzhang,liu2026jiuzhang4,shor1997polynomial,grover1996fast,lloyd1996universal,harrow2009quantum}. A credible advantage therefore needs both a proof and a measurement, on a workflow that runs on classical data, where input loading, read-out and strong classical baselines can each destroy a speedup~\cite{aaronson2015read,tang2019dequantization,holevo1973bounds,liu2021rigorous,cerezo2022challenges,biamonte2017quantum,lecun2015deep,jumper2021alphafold}.
The bridge to a quantum register is an encoding. Write $P_\mu^{(k)}$ for the discretised $k$-point marginal of the invariant measure. An efficiently preparable approximation to $P_\mu^{(k)}$ is carried by the Born statistics of a parametrised quantum state $\rho_\theta^{(k)}$, so that statistical questions about the marginal become expectation-value queries on a reusable quantum state. The physical hierarchy then acquires an operator grading: computational-basis statistics, and every linear functional of them such as covariances and correlation functions, lie in the diagonal Pauli-$Z$ sector, while phase-sensitive structures extend into the non-diagonal sectors containing Pauli $X$ and $Y$. Interrogating invariant structure becomes the problem of querying a quantum statistical memory: a compact quantum representation of a domain-specific invariant statistic, trained once on classical scientific data and interrogated for Pauli expectations selected after measurement. Quantum-informed machine learning (QIML)~\cite{wang2026qiml} supplies the framework, and the advantage therefore attaches to a statistical measure that science already requires.

The construction of Wang et al.~\cite{wang2026qiml} is the diagonal $k = 1$ member of this family, in which the extracted statistic is fixed before training. We index the family by the spatial order $k$ and by the Pauli class recovered at read-out, so that the stored statistics may be queried after acquisition rather than chosen in advance. Collective measurement theory~\cite{huang2022quantum} settles what can be extracted from an unknown state; the question here is which invariant statistics are worth extracting, judged by their effect on long-horizon prediction.

To answer this question, three stages establish the method (Fig.~\ref{fig:advantage}). Stage~\ref{res:representation} bounds the representation: a $k$-indexed family of quantum priors (Q-Priors) carries the $k$-point marginal of the invariant measure on $\nq = kq$ qubits, $q$ per site, with $\mathrm{poly}(\nq)$ parameters, against the $2^{\nq} - 1$ weights of an explicit table, conditional on efficient preparability. Stage~\ref{res:readout} bounds the access: joint Bell measurements on two copies estimate the magnitude of any \emph{post hoc} Pauli expectation to additive accuracy $\eta$ with failure probability $\delta$, using $M_Q = O(\eta^{-4} \log(1/\delta))$ copy pairs independent of $\nq$, whereas any adaptive single-copy protocol for the same task requires $M_C = \Omega(2^{\nq})$; the separation is established for the encompassing \emph{post hoc} state-access task using the measurement-tree framework of Huang et al.~\cite{huang2022quantum}, and we then instantiate its collective-access protocol on Q-Priors trained from scientific data. Stage~\ref{res:screening} turns that access into interrogation: because the Bell-measurement dataset is not consumed by the query that follows it, a candidate universe $\mathcal{U}$ of structures is answered from one acquisition at a copy cost growing only as $\log|\mathcal{U}|$, so the candidate space may be exponentially large while the number of state copies acquired stays polynomial in $\nq$. The collective protocol runs on IQM Garnet and Emerald superconducting processors (Fig.~\ref{fig:measure}), which places the separation within reach of present hardware rather than of a future generation.

Two case studies then interrogate a real invariant measure. Case~1 (turbulent channel flow) fixes what the non-diagonal read-out physically means: the Bell estimate is the velocity-direction coherence, a correlator of the turbulent invariant measure, and it extends to a balanced hierarchy of phase closures still strongly resolved at eighth order. Case~2 (ERA5 medium-range forecasting) carries the architecture into an operationally motivated forecasting workflow, where the diagonal $k \leq 2$ Q-Prior steers a Koopman rollout and holds it away from collapse onto a static climatological mean field, and then resolves a non-diagonal low-order sector through planetary-wave phase coherence. In the systems studied here, statistical richness extends to high order while predictive utility concentrates in a small set of complementary low-order sectors, so the memory holds and exposes a richer object than any one predictor consumes. Section~\ref{sec:discussion} develops this.

\begin{figure}[htbp]
\centering
\includegraphics[width=0.9\textwidth]{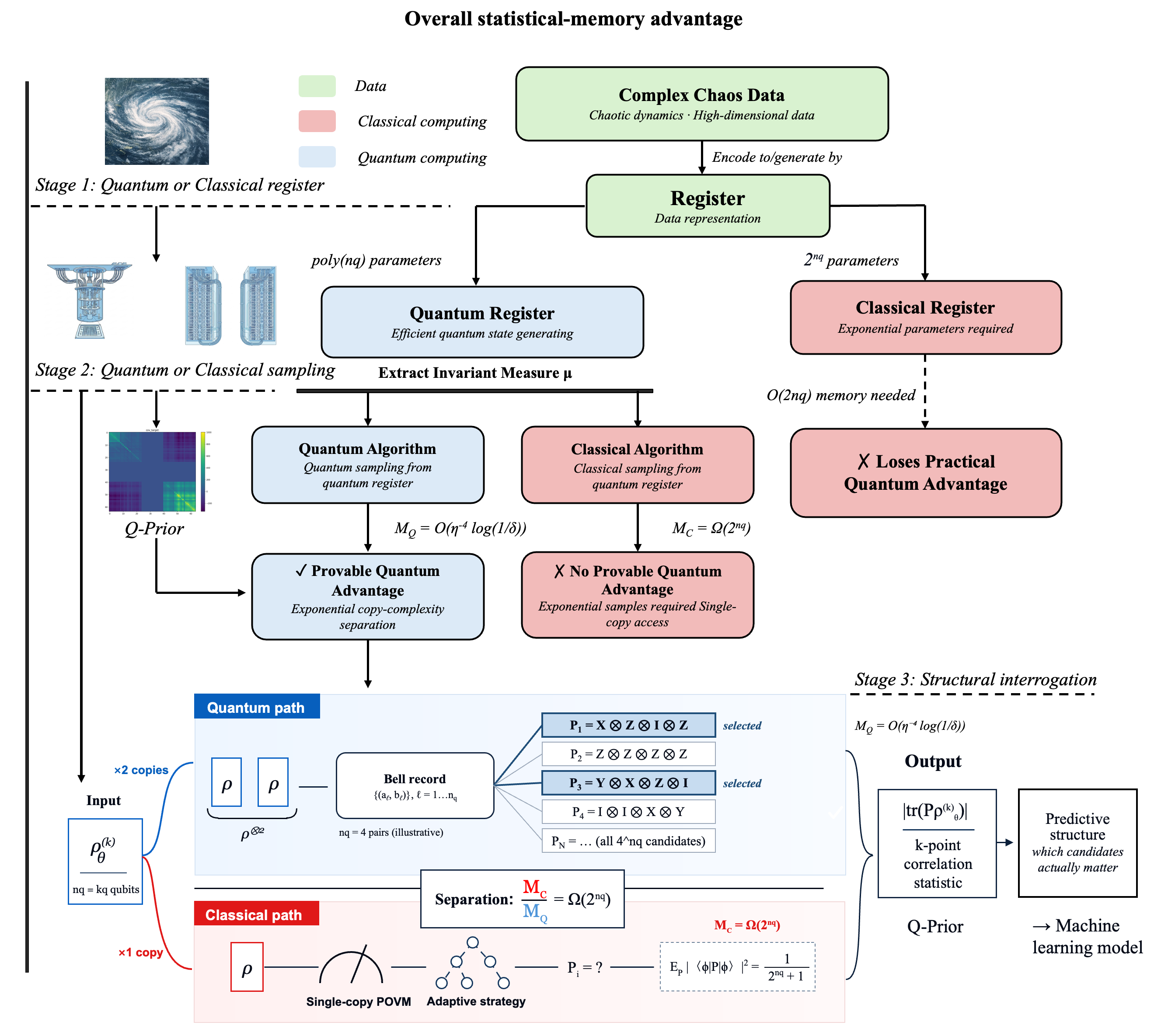}
\caption{Quantum statistical memory and its interrogation. Invariant structure enters at the left and predictive selection leaves at the right; colours mark data (green), classical computing (red) and quantum computing (blue). \emph{Stage 1, represent}: a quantum register holds the Q-Prior with $N_Q = \mathrm{poly}(\nq)$ trainable parameters. \emph{Stage 2, acquire}: invariant-measure statistics are extracted from the stored state, and Bell measurements on two copies of $\rho^{(k)}_\theta$ estimate \emph{post hoc} Pauli-expectation magnitudes with $M_Q = O(\eta^{-4}\log(1/\delta))$ copy pairs, while any adaptive single-copy protocol for the corresponding full-Pauli task requires $M_C = \Omega(2^{\nq})$. \emph{Stage 3, interrogate}: the resulting Bell-measurement dataset is a fixed classical object that controls a family of candidate Pauli observables uniformly, two members shown as selected, at an acquisition cost growing only as $\log|\mathcal{U}|$, so members may be chosen for analysis after the data are taken. The Bell dataset supports \emph{post hoc} analysis of candidate Pauli magnitudes. Signed or complex targets for the downstream predictor are obtained through the extraction procedures described in the case studies. The three stages together constitute the statistical-memory advantage, implemented as QIML~\cite{wang2026qiml}.}
\label{fig:advantage}
\end{figure}

\section{A statistical memory that can be interrogated}
We represent selected $k$-point structure of the invariant measure~\cite{eckmann1985,young2002srb} by a Q-Prior: a parametrised quantum state whose induced statistics are fitted to the corresponding physical data. It is a compact statistical memory for that marginal and for the sectors reachable from it, not a representation of the entire measure in one state. We treat it as a family indexed by the spatial order $k$ and by the Pauli class extracted at read-out, the construction of Wang et al.~\cite{wang2026qiml} being the diagonal $k = 1$ member. Training occurs once; the resulting state thereafter serves marginal and correlation statistics to the classical predictor, following the hybrid division of labour in which parametrised quantum subroutines are embedded in classical optimisation loops~\cite{bickley2025extending,tilly2022variational,kandala2017hardware}. Because the object measured is reusable, the classical fields need not be reloaded at inference, and the loading and read-out obstacles identified in the Introduction are removed by a single construction rather than addressed separately.

Throughout, $k$ denotes the Q-Prior spatial order, that is the number of locations represented jointly, and $n$ is reserved for the order of the turbulent phase-closure hierarchy of Section~\ref{sec:case1}. For $k$ spatial locations modelled jointly by the Q-Prior, each discretised into $B = 2^q$ bins on $q$ qubits, the generator acts on $\nq = kq$ qubits and prepares a pure state
\begin{equation}
\rho^{(k)}_\theta = U(\theta)|0^{\nq}\rangle\langle 0^{\nq}| U^\dagger(\theta), \qquad p_\theta(s) = |\langle s | U(\theta) | 0^{\nq}\rangle|^2 .
\label{eq:generator_state}
\end{equation}
Here $p_\theta(s)$ is the Born-rule probability of computational-basis outcome $s$; training adjusts $\theta$ from finite data so that chosen low-order marginals or moments of $\rho^{(k)}_\theta$ track targets. The learned object is this fitted generator state and its empirically constrained statistics. For projectors $\Pi_{i_l, b_l}$ ($l = 1, \ldots, k$) selecting bin $b_l$ at location $i_l$, the corresponding marginal is 
\begin{equation}
\mathcal{P}^{(k)}_\theta(b_1, \ldots, b_k) = \mathrm{tr}\!\left[\left(\bigotimes_{l=1}^{k} \Pi_{i_l, b_l}\right) \rho^{(k)}_\theta\right].
\label{eq:marginal}
\end{equation}

\paragraph*{Q-Prior as a quantum compressed statistical memory for invariant measure.}
A Q-Prior $p_\theta$ is best read as a parametrised approximation to a low-order marginal of the physical invariant measure $\mu$~\cite{eckmann1985}. The $k$-indexed family extends the single-site estimator~\cite{wang2026qiml} to multi-site joint marginals, reaching spatial correlations that the single-site case cannot resolve. This matters because pointwise predictability in a chaotic rollout decays on the Lyapunov time~\cite{eckmann1985}, while the rollout distribution still converges to $\mu$; fidelity to $\mu$ therefore survives as a long-horizon target long after individual trajectories have decorrelated~\cite{petersen1989ergodic,budivsic2012applied}. In ergodic-theory terms $\mu$ is the physical, Sinai--Ruelle--Bowen-type measure governing time averaged observables of dissipative chaotic flows~\cite{young2002srb}. The division of labour is clean here: the Koopman component handles the operator side of the dynamics, the Q-Prior the measure side, which is the fixed point of the dual Perron--Frobenius action.

Superposition and entanglement play distinct parts here, and the second is essential. Superposition supplies the state space: $\rho^{(k)}_\theta$ spans $2^{\nq}$ amplitudes, so the induced weights $p_\theta(s)$ can carry extensive support. But support alone buys nothing, since a product state over $\nq$ qubits already spans the same outcomes. What entanglement across the site partition supplies is non-factorisability, and that is the property a classical product representation cannot imitate at any parameter budget~\cite{benedetti2019generative}. The extracted statistics then enter a classical loss by the standard differentiable-programming route, in which a constraint is optimised through the downstream model~\cite{liu2018differentiable}.

The representation stage is bounded first.

\begin{result}[Compact Q-Prior representation]
\label{res:representation}
Consider the $k$-point marginal of the invariant measure $\mu$, with each site discretised into $B = 2^q$ bins. An explicit classical table of the full joint distribution requires
\[
N_C = B^k - 1 = 2^{kq} - 1
\]
independent parameters, while a fully factorised single-site product model requires
\[
N_{\mathrm{prod}} = k(B-1) = k(2^q - 1)
\]
parameters and carries no inter-site correlation. If this $k$-point marginal is preparable to fixed accuracy by a polynomial-size circuit on $\nq = kq$ qubits, then a Q-Prior hosts the corresponding non-factorisable joint distribution with
\[
N_Q = \mathrm{poly}(\nq)
\]
trainable parameters.
\end{result}

Within the pure-state Q-Prior ansatz, entanglement across the site partition supplies the non-factorisability required to represent the learned inter-site correlations; the preparability condition and the parameter counts are formalised in Supplementary Section~S1. The counts are measured, not merely asymptotic: at the case-study scale ($k = 2$, $q = 5$, $\nq = 10$), the trained generator uses $170$ parameters against the $1{,}023$ independent probabilities of the explicit table (Supplementary Section~S4).

Posing the interrogation as a Pauli-query problem is what makes it tractable. Each computational-basis projector on $q$ qubits expands in the diagonal Pauli subalgebra $\{I,Z\}^{\otimes q}$, so a $k$-fold tensor product of bin projectors expands into strings $P \in \{I,Z\}^{\otimes \nq}$. Every $k$-point marginal, and every linear functional of it such as a covariance, a correlation function or a higher moment, is therefore a linear combination of Pauli expectations $\mathrm{tr}(P\rho^{(k)}_\theta)$ (Supplementary Section~S1). Encodings that place a physical phase in a qubit phase extend the family beyond the diagonal sector into strings containing $X$ and $Y$, which is where the turbulent directional correlators of Section~\ref{sec:case1} live. Interrogating candidate invariant structures reduces to querying Pauli observables of the stored state, and the operative question becomes how cheaply those queries can be served.

Access to the stored state is formalised as a \emph{post hoc} Pauli task, in which the observable is selected after measurement. A trained Q-Prior state $\rho^{(k)}_\theta$ on $\nq=kq$ qubits is available in $M$ identical copies ($M_Q$ counts Bell copy pairs, $M_C$ single copies; $t$ is reserved for the forecast lead time in the weather sections). $M_C$ denotes the classical single-copy cost, provably $\Omega(2^{\nq})$ for the \emph{post hoc} full-Pauli task and realised at $3^{\nq}/\eta^2$ by classical shadows. Measurements are performed before the query is revealed. After this learning phase, a Pauli observable $P\in\{I,X,Y,Z\}^{\otimes \nq}$ is revealed, and the algorithm estimates $|\mathrm{tr}(P\rho^{(k)}_\theta)|$ to additive accuracy $\eta$ with failure probability at most $\delta$. The same trained state must therefore support statistical constraints chosen after training.

\begin{result}[Collective \emph{post hoc} Pauli access]
\label{res:readout}
For the \emph{post hoc} Pauli read-out task on an $\nq$-qubit state, joint Bell measurements on two copies estimate $|\mathrm{tr}(P\rho)|$ for any Pauli $P$ using
\begin{equation}
M_Q=O\!\left(\eta^{-4}\log\frac{1}{\delta}\right),
\label{eq:MQ_scaling}
\end{equation}
copy pairs, independent of $\nq$ at fixed accuracy and confidence. Any adaptive single-copy protocol for the corresponding full-Pauli read-out task requires $M_C=\Omega(2^{\nq})$ copies. In this work the upper-bound protocol is instantiated on trained Q-Prior states.
\end{result}

Each copy of $\rho_\theta$ is prepared on the device and destroyed when measured, so the separation is stated in copies. In the measurement-access model the learner receives identical copies of the trained state, holds bounded classical memory, and answers a Pauli query revealed only after measurement, without access to the circuit description $\theta$. In this setting joint two-copy Bell measurement estimates any $|\mathrm{tr}(P\rho)|$ at a copy count independent of $\nq$, whereas any adaptive single-copy strategy, measuring one copy at a time and retaining only classical memory between measurements, requires $\Omega(2^{\nq})$ copies. The $\nq = 16$ hardware points are consistent with this scaling on trained, data-derived states.

The lower bound is stated in the measurement-access model without access to the circuit description. In the trained-Q-Prior workflow the parameters may be known, but evaluating generic Pauli expectations from the circuit description is a separate classical-simulation problem (Supplementary Section~S2)~\cite{terhal2004adaptive,bremner2011classical}; the collective protocol instead reads these observables directly from the prepared memory.

The upper and lower bounds follow from the measurement-tree framework for the encompassing \emph{post hoc} state access problem~\cite{huang2022quantum}. We pair this separation with the representation stage of Stage~\ref{res:representation}, which has no analogue in the learning-from-unknown-states setting, and instantiate the collective protocol on a trained prior whose marginals approximate the invariant measure of a chaotic flow. A separation about learning black-box states becomes a constructive, reusable read-out method for a data derived scientific prior: train once, then estimate \emph{post hoc} the magnitude of any Pauli expectation and feed it to the downstream classical predictor. The scientific relevance is whether the statistics this channel can recover include physical objects that matter for stabilising long-horizon prediction.

The upper bound pairs corresponding qubits of $\rho^{\otimes 2}$ and measures each pair in the Bell basis. Every Bell outcome is a simultaneous eigenstate of $\sigma \otimes \sigma$ for $\sigma \in \{X,Y,Z\}$, so the product of Bell eigenvalues along a Pauli string is an unbiased estimator of $[\mathrm{tr}(P\rho)]^2$, and Hoeffding concentration with the square-root conversion gives Eq.~\eqref{eq:MQ_scaling}. Classical shadows~\cite{huang2020shadows} are the corresponding single-copy scheme, at $O(3^{\nq}/\eta^2)$ samples for arbitrary Paulis, and remain $\nq$-dependent.

For the lower bound, we compare against adaptive single-copy protocols. The hard family is centred at $\rho_0 = I/2^{\nq}$ with alternatives $\rho^{(s, P)} = (I + s\alpha P)/2^{\nq}$, where $P$ is a uniformly random non-identity Pauli, $s = \pm 1$, and $\alpha \in (0,1)$ is a fixed bias amplitude. The Pauli--SWAP identity gives any single-copy measurement an average squared sensitivity of $1/(2^{\nq}+1)$ to a uniformly random Pauli direction, and the resulting total-variation bound forces $M_C = \Omega(2^{\nq})$. Constants, construction and the measurement-tree proof are in Supplementary Sections~S1, S2 and~S3.

The mechanism behind this separation is geometric. The Bell basis is a maximally entangled measurement basis whose four projectors are simultaneous eigenstates of $\sigma \otimes \sigma$ for every Pauli operator $\sigma$. A single Bell-measurement dataset on two copies therefore carries the eigenvalue information for every Pauli string simultaneously, and can be reused by classical post-processing to estimate any $|\mathrm{tr}(P\rho)|$ after the query is revealed. A single-copy measurement instead commits to a direction before the query is known. The same $\sigma \otimes \sigma$ geometry underlies Bell correlation experiments~\cite{aspect1982}, here repurposed for collective access to the Pauli statistics of a trained prior.

What then matters is how large a candidate space can stay available for later analysis. The Bell-measurement dataset produced at Stage~\ref{res:readout} is not consumed by the query that follows it. The same Bell outcomes enter the estimator for every Pauli string, so a candidate universe specified independently of the realised outcomes can be controlled from one acquisition, with any subset subsequently selected for analysis. For a scientific search the operative cost is therefore not that of a single query but that of keeping a large candidate space open.

\begin{result}[Structural interrogation]
\label{res:screening}
Fix a candidate universe $\mathcal{U} \subseteq \{I,X,Y,Z\}^{\otimes \nq}$ of Pauli observables, specified independently of the realised measurement outcomes. A single Bell-measurement dataset on $M_Q$ copy pairs of $\rho$ returns estimates $\widehat{C}_P$ satisfying the uniform guarantee
\begin{equation}
\Pr\Bigl[\, \max_{P \in \mathcal{U}} \bigl|\widehat{C}_P - |\mathrm{tr}(P\rho)|\bigr| \leq \eta \,\Bigr] \geq 1 - \delta
\label{eq:screening}
\end{equation}
whenever
\begin{equation}
M_Q \;\geq\; \frac{2}{\eta^4} \ln\frac{2|\mathcal{U}|}{\delta} \;=\; O\!\left(\eta^{-4}\Bigl(\log|\mathcal{U}| + \log\frac{1}{\delta}\Bigr)\right).
\label{eq:MQ_screening}
\end{equation}
Because the guarantee is uniform over $\mathcal{U}$, any subset of $\mathcal{U}$ may be selected and analysed afterwards from the same fixed dataset. For the full Pauli universe, $|\mathcal{U}| = 4^{\nq}$ and $M_Q = O(\eta^{-4}(\nq + \log(1/\delta)))$, polynomial in $\nq$, while any adaptive single-copy protocol for the same \emph{post hoc} task requires $M_C = \Omega(2^{\nq})$.
\end{result}

An exponentially large candidate structure space can therefore remain open for later interrogation while the quantum acquisition cost stays polynomial in register size. The $\log|\mathcal{U}|$ scaling is on the copies consumed on the device; classical post-processing is separate and grows with the number of members actually evaluated. Proof and the role of fixing $\mathcal{U}$ in advance are given in Supplementary Section~S2.

Structural search is the simplest instance. When the state carries a single distinguished but initially unknown Pauli direction of strength $\alpha$, as in the hard family used for the lower bound, taking $\mathcal{U}$ to be the non-identity Paulis at $\eta = \alpha/3$ locates that direction with probability at least $1-\delta$ using
\begin{equation}
M_Q \;\geq\; \frac{162}{\alpha^4} \ln\frac{2(4^{\nq}-1)}{\delta} \;=\; O\!\left(\alpha^{-4}\Bigl(\nq + \log\frac{1}{\delta}\Bigr)\right)
\label{eq:identify}
\end{equation}
copy pairs, while any adaptive single-copy protocol that locates it inherits the $\Omega(2^{\nq})$ bound (Supplementary Section~S2). Locating an unknown structured direction is polynomial in $\nq$ for collective access and exponential for single-copy access.

The three stages combine into a single statement about the memory.

\begin{advantage}
Storage gives an exponential parameter reduction relative to explicit tabulation, conditional on efficient preparability,
\begin{equation}
\frac{N_C}{N_Q} = \frac{2^{kq}-1}{\mathrm{poly}(kq)} = 2^{\Omega(\nq)} ,
\label{eq:memory_advantage}
\end{equation}
and collective access gives an exponential copy-complexity separation against adaptive single-copy protocols, in two forms at fixed accuracy and confidence. For a single \emph{post hoc} query, $M_Q = O(\eta^{-4}\log(1/\delta))$ is independent of $\nq$ and $M_C/M_Q = \Omega(2^{\nq})$. For uniform interrogation of the full Pauli universe, $M_Q = O(\eta^{-4}(\nq + \log(1/\delta)))$ and
\begin{equation}
\frac{M_C}{M_Q} = \Omega\!\left(\frac{2^{\nq}}{\nq}\right) = 2^{\Omega(\nq)} ,
\label{eq:memory_advantage_family}
\end{equation}
so keeping an exponentially large candidate space open costs a factor $\nq$ while the separation stays exponential. Together, these representation and interrogation bounds constitute the quantum statistical-memory advantage, established for a general class of states rather than for a construction chosen to exhibit it.
\end{advantage}

The two ratios are benchmarked differently. The storage ratio $N_C/N_Q$ is measured against explicit tabulation and against compressed classical families at matched parameter budget on the same data: a diagonal Gaussian mixture ($169$ parameters), low-rank non-negative matrix factorisation ($128$ parameters) and an autoregressive bit model ($61$ parameters) reach a best relative covariance error of $3.74 \times 10^{-3}$, against $2.31 \times 10^{-4}$ for the Q-Prior, with the classical models competitive on the auxiliary distribution metrics (Supplementary Section~S4). The access ratio $M_C/M_Q$ needs no such comparison: the collective-access bound holds for any trained state in the measurement-access model. The trained state answers every \emph{post hoc} Pauli-expectation magnitude query at $\nq$-independent cost, with no further access to the training data.

\paragraph*{Interface to the downstream predictor.}
The stored statistics reach the classical model through a differentiable constraint. The weather case study uses the $k \leq 2$ members of the family through a covariance regulariser,
\begin{equation}
\mathcal{L}_{\mathrm{tot}} = \mathcal{L}_{\mathrm{rec}} + \lambda \, \| \widehat{\Sigma}_t - \Sigma_Q \|_F^2 ,
\label{eq:loss}
\end{equation}
where $\widehat{\Sigma}_t$ is the empirical covariance of the rollout state over a sliding window ending at forecast time $t$.
Equation~\eqref{eq:loss} enforces the dynamics-measure duality at the level of second moments: $\Sigma_Q$ is a quantum-compressed approximation to the covariance of $\mu$ computed from Q-Prior marginals, and the regulariser penalises drift of the rollout away from this target. Higher-order $k \geq 3$ Q-Priors access higher-order statistics of $\mu$ within the same protocol.

\paragraph*{Two access channels from one trained state.}
A trained Q-Prior answers queries through two channels, which the case studies use for different purposes. The \emph{diagonal channel} samples the prior in the computational basis, one copy at a time, and returns signed two-site joint statistics, among them the covariance target $\Sigma_Q$ consumed by the ERA5 workflow. No quantity here passes through $|\mathrm{tr}(P\rho)|$, so the sign is retained directly. The channel costs one copy per shot and reaches only the diagonal sector. The \emph{collective (Bell) channel} measures two copies jointly and returns \emph{post hoc} Pauli-expectation magnitudes; this is the channel on which the copy-measurement separation of Stage~\ref{res:readout} operates.

We validated the diagonal channel end to end on the ERA5 case-study generator ($\nq = 10$, circuit depth $L = 5$, $170$ parameters), measuring the full register in the computational basis and reconstructing the signed two-site statistic $T_Q = \Sigma_Q/(s_A s_B)$, the covariance target normalised by the marginal standard deviations $s_A$ and $s_B$ at the two sites. In noiseless simulation the pipeline returns $T_Q^{\mathrm{sim}} = 0.963$ against a data target $T_{\mathrm{data}} \approx 0.963$; on IQM Garnet the same generator gives $T_Q^{\mathrm{hw}} = 0.66 \pm 0.088$. Diagonal extraction requires the full-register correlated distribution of the depth-$5$ circuit to survive intact, whereas the Bell read-out appends only a shallow pairwise measurement (Supplementary Section~S8).

The collective channel reaches statistics the diagonal channel cannot, and those statistics include physical objects that matter for long-horizon prediction. Section~\ref{sec:case2} demonstrates this with an inter-band $Z_{500}$ wave-5 phase coherence, a quantity defined in meteorological terms and reached through the Bell channel but not through diagonal sampling. Constraining it raises the ERA5 anomaly correlation at 144\,h with no covariance prior present, and adds further long-lead skill when combined with the covariance target. The main forecast benchmark and the phase sector therefore supply different information, and the copy-complexity separation of Stage~\ref{res:readout} operates on the channel that reaches the second.

\begin{figure}[htbp]
\centering
\includegraphics[width=0.72\textwidth]{qiml_measure.pdf}
\caption{\emph{Post hoc} Bell read-out of trained Q-Priors in simulation and on superconducting hardware. (A)~The two IQM processors used: 20-qubit Garnet ($\nq \leq 10$) and 54-qubit Emerald ($\nq = 12$ to $16$); highlighted sites host the $2\nq$-qubit two-copy register. (B)~Bell estimates against exact $|\mathrm{tr}(P\rho)|$ for diagonal and non-diagonal Paulis on the ERA5-trained Q-Prior in noiseless simulation; shaded band $\pm\eta$. (C)~Representation stage. The Pearson correlations describe the full two-site joint at displacement $r = (0,4)$; the quantity the regulariser of Eq.~\eqref{eq:loss} consumes is the scalar covariance target $\Sigma_Q$, which matches the ERA5 $Z_{500}$ data to within $0.4\%$ for the $\nq = 6$ prior used on Garnet and to $0.05\%$ for the $\nq = 10$ sweep prior. (D)~Copy cost of \emph{post hoc} Pauli read-out: single-copy classical-shadow cost $M_C \sim 3^{\nq}/\eta^2$ against Bell copy-pair cost $M_Q$, in noiseless simulation (blue) and on IQM hardware (purple). As in Fig.~\ref{fig:advantage}, the plotted classical line is the classical-shadow comparator, which lies above the proved $\Omega(2^{\nq})$ bound. Read-out values, protocol and shot accounting in Supplementary Section~S4.}
\label{fig:measure}
\end{figure}

\paragraph*{Scope and conditions.}
Two properties of the classical system determine where the framework applies. The extraction stage (Stage~\ref{res:readout}) is general and assumes nothing about the underlying system. The representation stage (Stage~\ref{res:representation}) rests on structure rather than capacity: an $\nq$-qubit register already spans $2^{\nq}$ basis outcomes, so raw support size is never the issue. First, the invariant measure should carry non-factorisable spatial correlations, since it is these long-range correlations reaching the downstream learner that the quantum register is storing. Second, the low-order marginals should be efficiently preparable by a polynomial-resource parametrised circuit, so that the Q-Prior can be trained. Spatially extended chaotic systems, including atmospheric, oceanic, plasma and turbulent flows, provide natural settings in which these conditions can be tested. Stating them this way makes the framework's dependence on correlation range explicit and testable.

\paragraph*{Hardware demonstration in the non-diagonal read-out regime.}
Copy-complexity separations are ordinarily demonstrated on state families engineered to exhibit them; we test the scaling directly on trained Q-Priors, in the general non-diagonal regime for which Stage~\ref{res:readout} is stated. The generator described in Supplementary Section~S4 is trained on paired ERA5 $Z_{500}$ samples at a fixed displacement under a pairwise-only objective, distinct from the forecast prior used in the ERA5 case study of Section~\ref{sec:case2}; the covariance of the learned joint is matched to the data, and the trained state is then read-out by Bell measurements on two copies. Test observables include full-weight Paulis in $\{X,Y,Z\}^{\otimes \nq}$ together with low-weight cross-site terms, evaluated \emph{post hoc} from a single Bell-measurement dataset. We fix $\eta = 0.2$ and $\delta = 0.1$, with $\eta$ set above the median noise floor of $0.104$ measured for the ERA5 set in Supplementary Section~S8. The Bell copy-pair count $M_Q$ stays in the $10^3$ band across $\nq \in \{4, \ldots, 16\}$, with simulation medians of $500$ to $800$ across Pauli read-outs, consistent with the $\nq$-independence of Eq.~\eqref{eq:MQ_scaling}. The single-copy classical-shadow cost, by contrast, grows as $3^{\nq}/\eta^2$ and reaches $M_C \approx 1.1\times 10^9$ at $\nq = 16$, giving a median ratio $M_C/M_Q \approx 1.3\times 10^6$ in simulation (Fig.~\ref{fig:measure}D). The same protocol executed on IQM superconducting processors, Garnet (20 qubits) for $\nq \leq 10$ and Emerald (54 qubits) for $\nq = 12$ to $16$ (two-copy registers of up to 32 physical qubits), yields hardware copy-pair counts in the same flat band (medians $10^2$ to $1.6\times 10^3$); hardware accounting is given in Supplementary Section~S8. The extracted non-diagonal expectations are Pauli-expectation magnitudes of the ERA5-trained Q-Prior, so the predicted $\nq$-independent scaling holds on trained, data-derived quantum states. The turbulent channel-flow case of Section~\ref{sec:case1} then gives this sector a physical reading, first through the velocity-direction coherence and then through a balanced closure hierarchy that stays strongly resolved to eighth order.

\section{Interrogating a chaotic invariant measure}
\label{sec:cases}

\subsection{Case 1: Turbulent channel flow}
\label{sec:case1}

The theory above establishes how broadly the statistical memory can be interrogated. We next ask what structure is actually present in physical invariant measures and which parts of it matter for prediction. Turbulence answers the first. Here the non-diagonal sector of the read-out coincides with correlators of the velocity field that fluid dynamics already names and uses, on a system of independent scientific interest. To give the non-diagonal sector a direct physical interpretation, we use the velocity field, whose local orientation $\theta = \arg(v_x + i v_y)$ supplies an intrinsic phase. Encoding $\theta$ at each site in a qubit phase, $|\psi(\theta)\rangle = (|0\rangle + e^{i\theta}|1\rangle)/\sqrt{2}$, the off-diagonal read-out $\langle \sigma^+_A \sigma^-_B \rangle = \tfrac{1}{4} \langle e^{i(\theta_A - \theta_B)} \rangle$, with $\sigma^\pm = (X \pm iY)/2$ the standard raising and lowering combinations of the Pauli operators on sites $A$ and $B$, equals the two-point velocity-direction coherence, a correlator of the turbulent invariant measure invisible to the diagonal magnitude statistics.

Two roles are separated in what follows. Direct phase encoding identifies physically meaningful observables inside the collectively accessible sector and measures what they cost to serve. Trained generative Q-Priors carrying the corresponding statistics then test the downstream predictive value of those observables in the Koopman rollout. We instantiate the first role on turbulent channel-flow data~\cite{wang2026qiml}. After subtracting the mean wall-normal profile, the directional coherence $C(r)$ is large at small streamwise separation and decays towards the decorrelated control level at large $r$ (Fig.~\ref{fig:tcf1}B). The two-copy Bell read-out recovers $C(r)$ within $\eta$ across the measured range (Fig.~\ref{fig:tcf1}D). On the same two-copy records, the accompanying \emph{post hoc} Pauli test set exhibits the predicted $\nq$-independent Bell cost against the exponentially growing single-copy comparator, in simulation and on the IQM Garnet ($\nq \leq 10$) and Emerald ($\nq = 12$ to $16$) processors (Fig.~\ref{fig:tcf1}C). Building on the channel-flow benchmark of Ref.~\cite{wang2026qiml}, we use the directional sector to probe progressively higher-order invariant structure and its downstream predictive value.

\paragraph*{From two-point coherence to a hierarchy of phase closures.}
The directional coherence has an $n$-site generalisation, which we evaluate directly on the direct numerical simulation (DNS) ensemble. For $n$ sites carrying velocity-direction phases $\theta_j$ and signs $\varepsilon_j \in \{+1, -1\}$, define the balanced phase closure of the flow,
\begin{equation}
C_n^{\mathrm{phys}} = \mathbb{E}_{\mathrm{DNS}}\!\left[ \exp\!\Big( i \sum_{j=1}^{n} \varepsilon_j \theta_j \Big) \right],
\label{eq:closure}
\end{equation}
where the expectation runs over the DNS ensemble and, for even $n$, the signs are balanced, $\sum_j \varepsilon_j = 0$, for instance $(+,+,-,-)$ at $n = 4$. Balancing removes the response to any uniform phase offset, so $C_n^{\mathrm{phys}}$ reports genuine multi-site organisation rather than single-site phase statistics. Later, $C_n^Q$ denotes the same closure family evaluated on a trained generative Q-Prior, which is the object the downstream experiment consumes; the two coincide only to the fidelity of the prior and we keep them distinct throughout.

The closures remain large deep into the hierarchy: $C_4^{\mathrm{phys}} = 0.882$, $C_6^{\mathrm{phys}} = 0.650$ and $C_8^{\mathrm{phys}} = 0.425$. We evaluate them on the full-resolution plane at wall-normal station $y = 48$ across separations $r \in \{2, 4, 8\}$, discarding the lowest fifth of speeds where the direction angle is ill-conditioned (Supplementary Section~S5). Two controls fix the origin of the signal. Scrambling the phases site by site leaves every single-site distribution intact but destroys inter-site organisation, and collapses all orders to $0.005$--$0.008$. The odd orders, which admit no balanced sign pattern, give $C_3^{\mathrm{phys}} = 0.100$ and $C_5^{\mathrm{phys}} = 0.125$.

The eighth-order closure remains $0.425$, roughly seventy times the scrambled level, showing that coherent multi-site phase structure survives far beyond pairwise order. Turbulent invariant measures are ordinarily characterised through low-order moments and spectra; the closure hierarchy resolves structure well above that level, and a quantum register can be made to report it.

\paragraph*{Collective access across the closure hierarchy.}
These closures are collectively accessible. Encoding each DNS phase into a qubit as above and reading out through the operator $M_{\bm\varepsilon} = \bigotimes_{j=1}^{n} \sigma_j^{\varepsilon_j}$, with $\sigma^{+} = (X + iY)/2$ and $\sigma^{-} = (X - iY)/2$ as in the two-point case, gives $\mathrm{tr}(\rho\, M_{\bm\varepsilon}) = 2^{-n} C_n^{\mathrm{phys}}$ on the encoded state $\rho$. Each site occupies one qubit, so the closure register carries $n = \nq$ qubits per copy and the comparators below are the $\nq$-dependent ones of Stage~\ref{res:readout} written at $n$. The two-copy Bell sector quantity is then
\begin{equation}
S_n = \bigl|\mathrm{tr}(\rho\, M_{\bm\varepsilon})\bigr|^2 = \frac{|C_n^{\mathrm{phys}}|^2}{4^n}, \qquad |C_n^{\mathrm{phys}}| = 2^n \sqrt{S_n}.
\label{eq:closure-bell}
\end{equation}
The closures span a wide range of magnitudes, so we resolve them to a fixed \emph{relative} tolerance, $|\widehat{C}_n - C_n^{\mathrm{phys}}| \leq \eta |C_n^{\mathrm{phys}}|$ at $\eta = 0.2$ and $\delta = 0.1$. Wherever $|C_n^{\mathrm{phys}}| < 1$ this is the stricter of the two targets, and we keep it separate from the additive criterion of Stage~\ref{res:readout} throughout. The Bell copy budgets are then $M_Q \leq 50$ at $n = 4$ and $n = 6$, both grid-floor upper bounds since the smallest tested shot count already sufficed, and $M_Q = 200$ at $n = 8$. What sets the budget is the closure magnitude, not the Pauli weight: a smaller $|C_n^{\mathrm{phys}}|$ tightens the absolute tolerance, while the eighth-order closure, a weight-eight observable, is still served at $200$ copy pairs. Going to high-order does not by itself make the state expensive to query. At the same relative tolerance the single-copy comparator is $M_C^{\mathrm{shadow}} \sim 3^n / (\eta |C_n^{\mathrm{phys}}|)^2$. On IQM Garnet we serve the $n = 4$ closure on an eight-qubit two-copy register at median $M_Q = 400$.

Collective access therefore reaches a physical observable of the turbulent invariant measure at genuinely high statistical order, on hardware, at modest copy cost. The downstream experiment is a separate step: there the same closure family is carried by a trained generative Q-Prior rather than by direct encoding, and the constraint uses the complex closure $C_n^Q$, whose real and imaginary parts we estimate by sampling the prior,
\begin{equation}
\widehat{C}_n^Q = \frac{1}{N} \sum_{s=1}^{N} \exp\!\Big( i \sum_{j=1}^{n} \varepsilon_j \theta_j^{(s)} \Big),
\label{eq:closure-sample}
\end{equation}
with $\theta^{(s)}$ drawn from the trained prior. The closure family thus appears twice in this case study, in the two roles the architecture separates. Direct phase encoding puts $C_n^{\mathrm{phys}}$ within reach of the collective Bell channel and fixes the physical identity and the serving cost of the observable. The generative Q-Prior then represents the same family as $C_n^Q$ and feeds it, through the single-copy sampling channel, to the downstream predictor. The three quantities are related by a fixed chain: the Bell channel returns the magnitude $|\mathrm{tr}(\rho M_{\bm\varepsilon})| = |C_n^{\mathrm{phys}}|/2^n$ of Eq.~\eqref{eq:closure-bell}, the sign and phase are not carried by that magnitude, and the downstream constraint therefore uses the complex closure $C_n^Q$ of Eq.~\eqref{eq:closure-sample}, whose real and imaginary parts are estimated from single-copy samples of the trained prior. Magnitudes establish that the structure is present and cheap to reach; the complex values are what the predictor consumes.

\begin{figure}[htbp]
\centering
\includegraphics[width=1\textwidth]{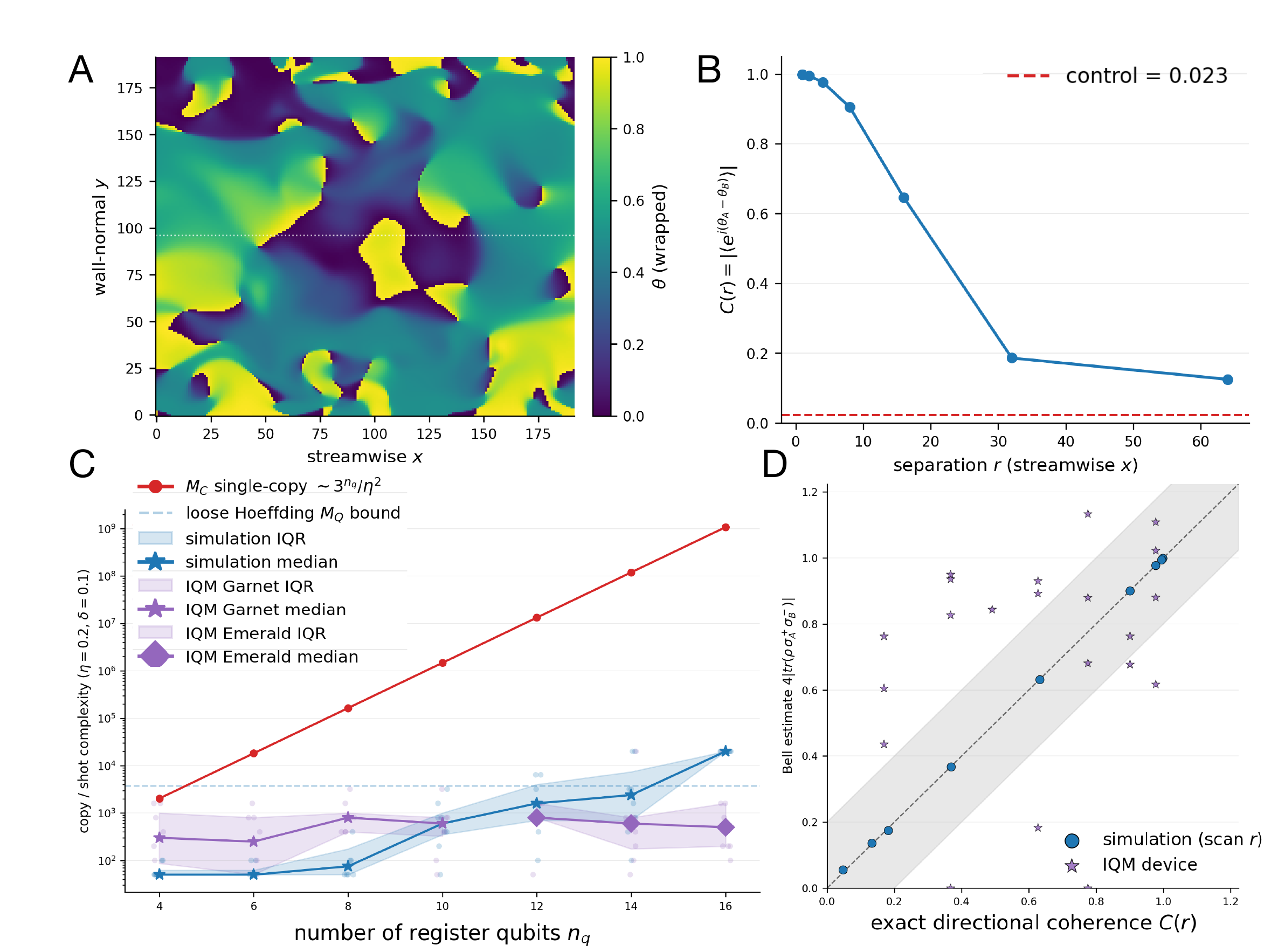}
\caption{QIML on turbulent channel flow: read-out mechanism. (A)~Velocity-direction phase field $\theta = \arg(v_x + i v_y)$ of a representative snapshot, encoded one site per qubit. (B)~Directional coherence $C(r) = |\langle e^{i(\theta_A - \theta_B)} \rangle|$ versus streamwise separation after subtracting the mean wall-normal profile; $C$ decays from near unity at small $r$ towards the decorrelated control level (dashed). (C)~Bell read-out scaling: copy-pair cost $M_Q$ over the \emph{post hoc} Pauli test set, in simulation (blue) and on IQM Garnet and Emerald hardware (purple), against the single-copy classical-shadow cost $M_C \sim 3^{\nq}/\eta^2$ (red). The directional coherence is a fixed two-local functional extracted from the same two-copy records and provides a physically interpretable member of the non-diagonal sector; the accompanying Pauli ensemble in this panel tests the generic full-Pauli scaling of Stage~\ref{res:readout}. (D)~Bell two-copy estimate $4|\operatorname{tr}(\rho \sigma^+_A \sigma^-_B)|$ against the exact $C(r)$ across separations, in simulation (blue) and on hardware (purple). Large-coherence values are recovered within $\eta$; small values are noise-floor limited. Protocol and noise accounting in Supplementary Sections~S5 and~S8.}
\label{fig:tcf1}
\end{figure}

\paragraph*{Downstream value of the low-order sectors.}
We next ask what the stored statistics are worth to a predictor, beginning with the low-order sectors. Two further generative Q-Priors are trained on the velocity statistics, a single-site ($k = 1$) and a multi-site ($k = 1 + 2$), and used as soft constraints on a Koopman rollout in a three-arm ablation that isolates the contribution of the multi-site $k = 2$ structure. Arm~(a) uses no Q-Prior; arm~(b) uses a single-site $k = 1$ Q-Prior that constrains only the rolled-out global velocity probability density function (PDF); arm~(c) uses the full $k = 1 + 2$ Q-Prior, which also constrains the rolled-out two-point joint distribution at streamwise separation $r = 8$. The second-order constraint is enabled partway through training. The first-order arm reaches $89.5\%$ held-out one-step agreement at convergence, and enabling the second-order constraint raises this to $93.4\%$ (Supplementary Table~S7), a $+3.9$ percentage point (pp) gain attributable to the multi-site structure. Across the same run the trajectory spans $+12.3$~pp, from $81.1\%$ where the $k = 1$ checkpoint is loaded to $93.4\%$ at convergence (Fig.~\ref{fig:tcf2}G,H); these three values are milestones of one training run, not a comparison against arm~(a). The rolled-out fields recover invariant-measure statistics across all four diagnostics examined.

\begin{figure}[htbp]
\centering
\includegraphics[width=0.9\textwidth]{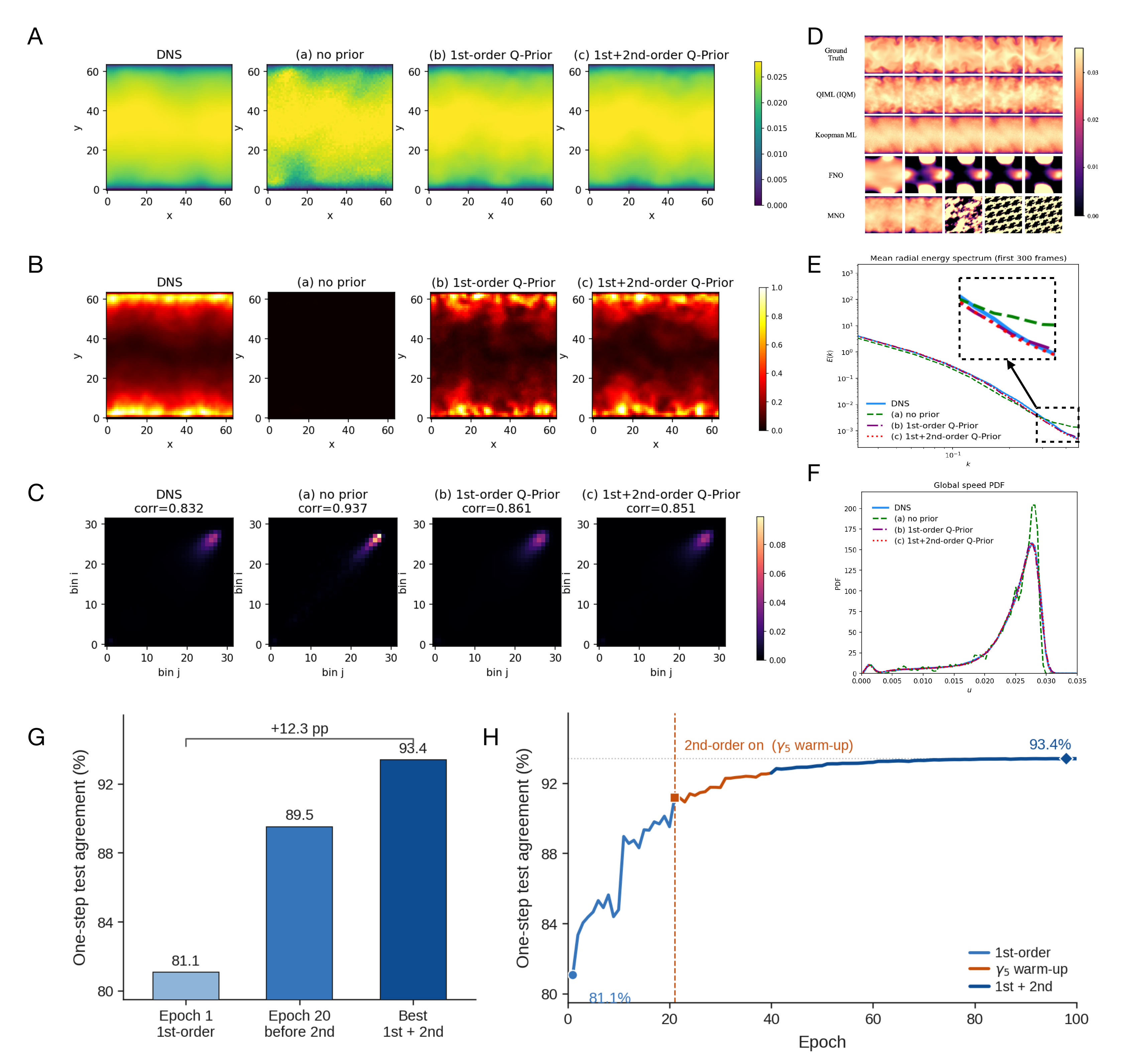}
\caption{QIML downstream performance on turbulent channel flow: three-arm ablation of the $k$-indexed Q-Prior structure. The Koopman-rollout arms are (a)~no Q-Prior, (b)~single-site ($k = 1$, first order) Q-Prior on the rolled-out velocity PDF, and (c)~$k = 1 + 2$ Q-Prior, which additionally constrains the two-point joint at streamwise separation $r = 8$. (A)~Time-averaged streamwise velocity. (B)~Time-averaged turbulent kinetic energy; arm~(a) collapses to a near-laminar state, while arms~(b) and (c) recover the DNS wall-region structure. (C)~Two-point joint at $r = 8$; the elevated correlation of arm~(a) reflects its near-laminar collapse rather than a closer match to the DNS pair structure. (D)~Rollout snapshots in Lyapunov-time units; the ground-truth, Koopman-only, FNO and MNO rollouts are from~\cite{wang2026qiml}, only the QIML rollout is computed here. (E)~Mean radial energy spectrum. (F)~Global speed PDF. (G)~Held-out one-step agreement for arm~(c) at three training milestones. (H)~The corresponding training curve, with the epoch at which the second-order term is switched on. Field statistics and joint-correlation values in Supplementary Section~S5.}
\label{fig:tcf2}
\end{figure}

The time-averaged streamwise velocity and the spatial distribution of turbulent kinetic energy are restored from a near-laminar baseline collapse to the DNS wall-region structure (Fig.~\ref{fig:tcf2}A,B). The radial energy spectrum tracks DNS, with the log root-mean-square error (RMSE) dropping from $0.20$ for the unregularised baseline to $0.05$ for the $k = 1 + 2$ prior, and the global speed PDF matches DNS with the Kolmogorov--Smirnov (KS) distance reduced from $0.06$ to $0.003$ (Fig.~\ref{fig:tcf2}E,F). Two-point pair structure is likewise closest to the DNS reference for the regularised arms; the unregularised arm overshoots because its rollout collapses onto an over-correlated, near-laminar field (Fig.~\ref{fig:tcf2}C, values in Supplementary Section~S5). Rollout snapshots confirm the same picture: the QIML rollout preserves turbulent structure, whereas the Koopman-only, Fourier neural operator (FNO)~\cite{li2020fourier} and Markov neural operator (MNO)~\cite{li2022learning} baselines drift or collapse to static or low-rank fields (Fig.~\ref{fig:tcf2}D). The channel-flow baselines follow Ref.~\cite{wang2026qiml}; the ERA5 study below replaces MNO with the adaptive Fourier neural operator (AFNO)~\cite{guibas2021afno}, matching the operator-learning baselines of that benchmark.

\paragraph*{Carrying the closure hierarchy through to prediction.}
Having established that the high-order closures are physically strong and collectively accessible, we carried the same hierarchy through the full chain: turbulent phase statistics, generative Q-Prior, the closure family $C_4^Q, C_6^Q, C_8^Q$, and a fourth Koopman arm~(d) that adds these closures to the pairwise constraint of arm~(c). The downstream calculation is simulator-based (Supplementary Section~S5); the hardware run above covers the serving step. Arm~(d) reaches $93.63 \pm 0.95\%$ held-out agreement across three runs ($93.6$, $94.6$ and $92.7\%$), against the $93.4\%$ of the pairwise arm~(c) under the same protocol: a mean shift of about $+0.2$~pp, well inside the run-to-run spread.

The high-order statistical family therefore passes cleanly through the complete memory-to-predictor chain, yet leaves the downstream metric essentially where the pairwise constraint had already put it. The dominant improvement occurs before the fourth-order sector is introduced. Nor is the response monotonic in order: fourth-order information gives a small positive increment, sixth-order regularisation can turn mildly detrimental, and eighth-order information again helps slightly. Once pairwise invariant structure is in place, the higher-order constraints appear largely redundant for this objective, and may compete with it for finite optimisation capacity.

The turbulent invariant measure is statistically rich well beyond pairwise order; the hierarchy remains accessible at modest Bell copy cost, with its collective-access mechanism demonstrated on hardware. Yet the portion of it that is actionable for stabilising the learned dynamics sits at low order. Case~2 refines the same picture in a global atmospheric workflow, showing that low-order predictive information can itself be distributed across complementary statistical sectors, and we return to the combined observation in Section~\ref{sec:discussion}.

\subsection{Case 2: ERA5 medium-range forecast}
\label{sec:case2}

\paragraph*{An operationally motivated forecasting workflow.}
Turbulence establishes that the accessible sectors carry real physical structure and that the structure persists to high order. Which parts of it a predictor actually uses is a different question, and ERA5 answers it in three steps: whether invariant-statistics regularisation stabilises the rollout at all, how much is gained by supplying richer information at the same spatial order, and whether a complementary non-diagonal low-order sector carries predictive information of its own. Medium-range weather forecasting is the sharper test: the classical baselines are strong, and the characteristic failure is exactly the one an invariant measure diagnoses. Iterated forecasts drift towards a static mean field~\cite{schiff2024dyslim}, losing the temporal structure of the atmosphere while their pointwise error scores look unremarkable. Regularising chaotic forecasts towards an invariant measure is established practice (DySLIM~\cite{schiff2024dyslim} for instance); what we add is the quantum prior that supplies the target, a compact reusable representation ($\mathcal{O}(10^2)$ parameters, Supplementary Table~S14) whose read-out spans the full-Pauli spectrum and so reaches non-diagonal structure a covariance only regulariser cannot (Stage~\ref{res:readout}; Section~\ref{sec:case1}). The case study uses the 500\,hPa geopotential height field ($Z_{500}$) from the European Centre for Medium-Range Weather Forecasts (ECMWF) ERA5 reanalysis~\cite{hersbach2020era5}, an observation anchored record of global atmospheric state used as the training reference of modern data-driven weather forecasting~\cite{lam2023learning,bi2023accurate,chen2023fengwu,pathak2022fourcastnet} and the basis of the WeatherBench medium range benchmark~\cite{rasp2020weatherbench}. Fields are represented at $1.5^\circ$ resolution. The training period is 1979 to 2015 and the held out test period is 2016 to 2017. Models are trained for 24\,h-ahead one-input one-output recurrent prediction and evaluated by iterated rollouts to 240\,h, with auxiliary 480\,h collapse diagnostics reported in Supplementary Section~S7. The baselines are a Koopman autoregressive model without Q-Prior, the FNO and the AFNO, all trained and evaluated under the same recurrent rollout protocol. We compare against operator learning baselines only: foundation models such as GraphCast, Pangu Weather, FengWu and FourCastNet run at higher resolution over a broader set of variables, which puts them outside a matched comparison here. All models use a single-step recurrent rollout protocol; the neural-operator baselines use the same single-step training protocol as QIML, ensuring a matched comparison; multi-step fine-tuning~\cite{lam2023learning} is orthogonal and could be applied to any of the models.

\begin{figure}[htbp]
\centering
\includegraphics[width=1 \textwidth]{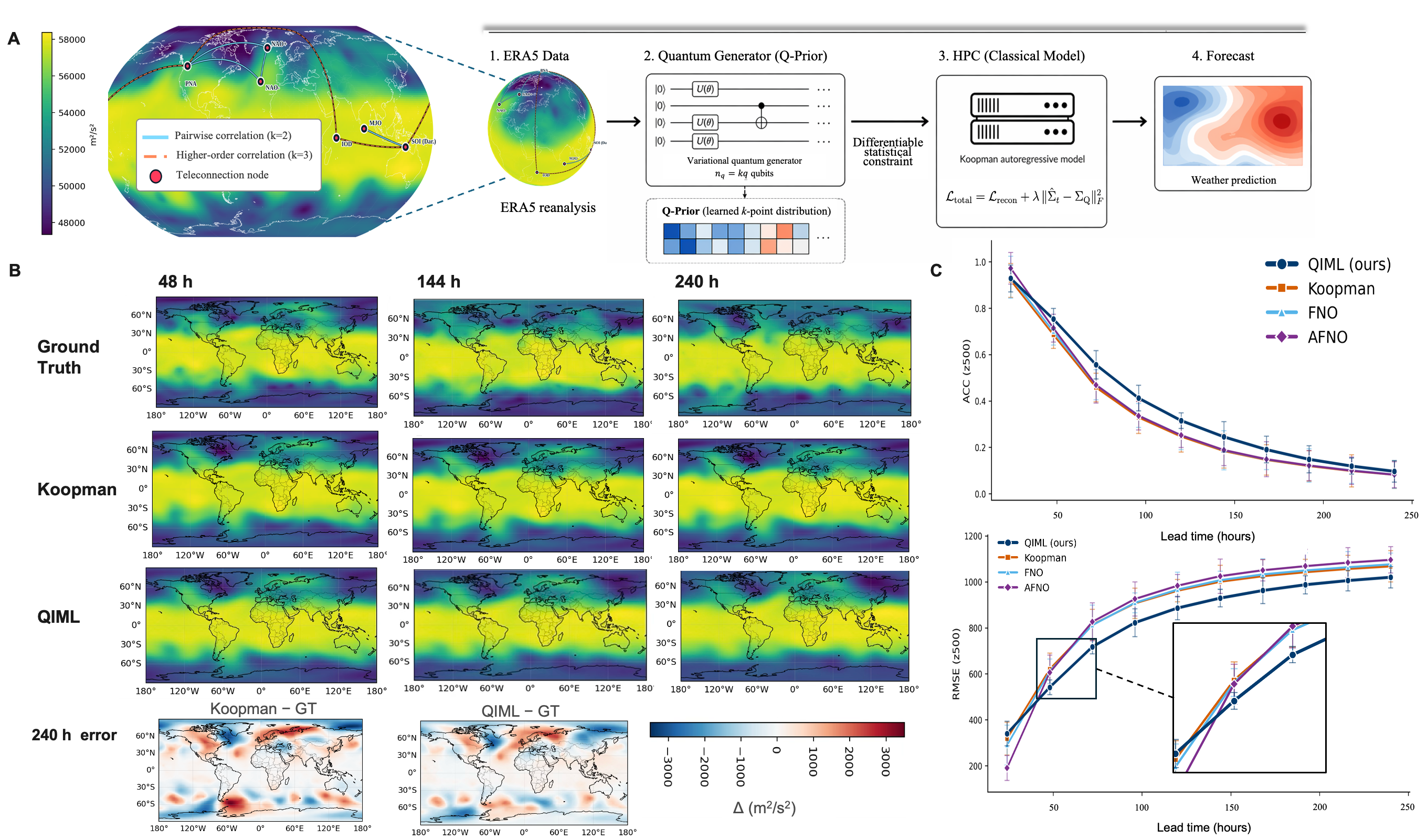}
\caption{QIML applied to ERA5 $Z_{500}$ forecasting. (A)~Hybrid workflow. A variational quantum generator on $\nq = kq$ qubits is trained from ERA5 samples so that its induced $k$-point Q-Prior statistics match the data; the extracted invariant statistical constraint enters a classical Koopman autoregressive model through the differentiable term of Eq.~\eqref{eq:loss}, which produces the forecast. The map indicates the pairwise ($k = 2$) and higher-order ($k = 3$) correlation structure the prior family represents, with teleconnection nodes shown for orientation; higher-order read-out is demonstrated in the channel-flow case study. The reported benchmark uses the diagonal covariance target $\Sigma_Q$, and the non-diagonal sector is treated separately in Fig.~\ref{fig:phase}. The forecast pipeline is classical; the Q-Prior enters only through the extracted statistics. (B)~Rollouts at 48, 144 and 240\,h for ground truth, the unregularised Koopman model and QIML, with signed 240\,h difference maps against ground truth below. The Koopman rollout drifts progressively from the observed large-scale structure, while the QIML rollout retains it. (C)~ACC (upper) and RMSE (lower) against lead time for QIML, Koopman, FNO and AFNO; the inset enlarges the short-lead region. The QIML advantage grows at medium and long lead times.}
\label{fig:weather}
\end{figure}

The reported case study uses the single-site and pairwise $k \leq 2$ Q-Prior, supplying diagonal projector marginals (equivalently, expectations of $Z$ type Pauli observables) that determine the covariance target $\Sigma_Q$ in Eq.~\eqref{eq:loss}. We take covariance as the starting observable rather than as a privileged one. As the lowest-order signed statistic of the invariant measure used here, it provides a minimal probe of attractor drift; only after establishing its stabilising role do we turn to the broader non-diagonal sector accessible to collective read-out. Stage~\ref{res:readout} establishes the read-out separation for the general Pauli class, including the higher-$k$ extensions; the main forecast benchmark instantiates the framework at the diagonal pairwise level, and the planetary-wave phase experiment below tests the non-diagonal sector.
The Q-Prior improves the anomaly correlation coefficient (ACC) of the Koopman model by 10 to 39\% across 48 to 240\,h lead times and reduces the 48\,h RMSE by about 13\%. In the results summarised in Fig.~\ref{fig:weather}C, it outperforms FNO and AFNO at every lead time beyond 24\,h. At short lead times, neural operator baselines remain competitive because the forecast remains close to the initial condition and errors are dominated by local, high frequency structure. At medium lead times, recurrent drift and mismatch with invariant statistics become increasingly important, and the Q-Prior constraint helps preserve the rollout distribution. At the longest leads the Q-Prior does not extend deterministic predictability; it suppresses drift towards a static mean field and keeps the forecast field varying coherently, which is what the invariant-measure interpretation predicts and what the rollout comparison shows (Fig.~\ref{fig:weather}B). The long-horizon contribution is rollout stability and invariant-measure fidelity rather than pointwise skill: at 240\,h the Q-Prior reaches ACC $0.114$ against $0.082$ for the unregularised Koopman baseline (Supplementary Table~S12). A scale-matched classical-covariance control traces the gain to invariant-statistics regularisation (Supplementary Section~S6). The Q-Prior supplies the invariant target from a compact reusable state whose collective read-out extends beyond the diagonal sector at $\nq$-independent copy cost.

\paragraph*{How much statistical information the forecast uses.}
The channel-flow closure hierarchy raised the question of how much of the available invariant-measure information a downstream predictor actually consumes. ERA5 admits a clean test of the same question at fixed spatial order. The regulariser of Eq.~\eqref{eq:loss} uses a single scalar summary of the $k = 2$ Q-Prior, the covariance target $\Sigma_Q$; the trained prior also holds the complete pairwise joint distribution, which carries substantially more information about the same two sites. Replacing the scalar target by the full pairwise joint therefore varies statistical information content while holding spatial order fixed, and is not a higher multi point ($k > 2$) experiment.

Under a matched reduced-budget protocol with a single seed and model selection on validation ACC at 144\,h ($\lambda^\star = 1$), the anomaly correlation at 144\,h is $0.3177$ without statistical regularisation, $0.3255$ with the scalar second-order Q-Prior target, and $0.3294$ with the complete pairwise joint (Supplementary Table~S10; Supplementary Section~S6). The full joint therefore adds $0.0039$ over its scalar summary, about $1.2\%$ relative, against the substantially larger step of $0.0078$ from introducing second-order regularisation at all. Other lead times move in the same direction (Supplementary Section~S6).

Adding information within a fixed diagonal sector therefore returns little. The next experiment tests a different low-order sector.

\paragraph*{Complementary low-order phase information.}
Covariance fidelity does not exhaust low-order invariant information. The mid-latitude circulation organises itself into planetary waves whose meridional phase relationship is a physical property of the atmospheric attractor, and which the scalar covariance target used here does not constrain. Writing the zonal Fourier coefficient of the geopotential anomaly $Z'$ for latitude band $B$ as
\begin{equation}
A_m^B(t) = \sum_{\phi \in B} w_\phi \frac{1}{N_\lambda} \sum_\lambda Z'(\phi, \lambda, t)\, e^{-im\lambda}, \qquad w_\phi = \cos\phi,
\label{eq:zonal-fourier}
\end{equation}
with latitude $\phi$, longitude $\lambda$, $N_\lambda$ grid longitudes, zonal wavenumber $m$ and $\theta_m^B(t) = \arg A_m^B(t)$, we select the inter-band wave-5 phase coherence
\begin{equation}
R = \Bigl\langle \cos\bigl[\theta_5^{45^\circ\text{--}60^\circ\mathrm{N}} - \theta_5^{60^\circ\text{--}75^\circ\mathrm{N}}\bigr] \Bigr\rangle .
\label{eq:wave5}
\end{equation}
where the average runs over the evaluation window. This is a non-diagonal statistic in the sense of Section~\ref{sec:case1}: it depends on phase differences rather than on magnitudes, and it lies in the sector the collective Bell channel reaches. It is also defined entirely in meteorological terms, independently of any quantum representation.

We test its downstream value by extending the rollout objective,
\begin{equation}
\mathcal{L}_{\mathrm{tot}} = \mathcal{L}_{\mathrm{existing}} + \lambda_R \bigl(R_{\mathrm{rollout}} - R_{\mathrm{target}}\bigr)^2 ,
\label{eq:phase-loss}
\end{equation}
where $R_{\mathrm{rollout}}$ is evaluated on the rolled-out field and $R_{\mathrm{target}}$ is the invariant reference value computed on the training period, in a matched continuation diagnostic with four arms: no structural prior ($\varnothing$), the phase constraint alone ($R$), the diagonal covariance prior alone ($\Sigma$), and both together ($\Sigma + R$). For the $R$-only arm $\mathcal{L}_{\mathrm{existing}}$ carries no covariance term; for $\Sigma + R$ it carries the covariance prior of Eq.~\eqref{eq:loss}. All arms continue from a common checkpoint under an identical protocol at weak coupling $\lambda_R = 0.1$, which is the setting reported here rather than a general optimum (Supplementary Section~S6).

The phase sector carries predictive value on its own (Fig.~\ref{fig:phase}). Used alone, $R$ raises the anomaly correlation at 144\,h from $0.2595$ to $0.2740$, without any covariance constraint. Added to the covariance prior it raises ACC at 144\,h from $0.2694$ to $0.2817$, and at 240\,h from $0.1366$ to $0.1471$, with the 240\,h RMSE falling from $1065.4$ to $1012.4$\,m. The two sectors are not interchangeable. Within this diagnostic the covariance constraint alone does not reproduce the selected phase target, leaving $|R_{\mathrm{rollout}} - R_{\mathrm{target}}|$ at $0.184$ and $0.214$ at 144 and 240\,h, whereas the combined constraint restores it to $0.038$ and $0.008$. The scalar covariance target is thus useful without being a complete description of the actionable low-order sector: it carries the long-range stabilisation established in the full-budget benchmark above, while the phase sector supplies information it does not contain. The combined arm gives both the best long-lead skill and the closest approach to the invariant phase target among the four arms. The $R$-only arm does not improve every metric uniformly, its RMSE at 144\,h being marginally worse than the unconstrained arm, so the claim we draw is about complementary long-horizon information rather than uniform gain.

Diagonal covariance and non-diagonal wave phase thus behave as complementary low-order invariant sectors, each supplying information the other does not. The same wave-5 observable is compatible with the collective read-out: under the Quantinuum H1 device-noise model on the Nexus platform~\cite{QuantinuumFAQ2024,Quantinuum3nines2024,QuantinuumEmulator2024}, the corresponding Bell measurement retains a normalised $L_2$ recovery above $0.80$ relative to the ideal result, with the constraint itself applied classically in the present pipeline (Supplementary Section~S8).

Beyond the forecast performance, the hardware footprint of the extraction stage at the case-study scale is modest. It requires two copies of the trained generator and Bell measurements between corresponding qubits. For the weather-relevant range $k = 2$ to $3$ and $q = 5$, this uses $2 \nq = 20$ to $30$ physical qubits. The hardware target is close to leading near-term superconducting platforms operating near $\sim 99.9\%$ two qubit fidelity~\cite{Preskill_2018,kjaergaard2020superconducting,kandala2019error,giurgica2020digital}. The Garnet and Emerald runs demonstrate the Bell read-out primitive at this scale; the forecast pipeline itself is simulator based in the present case study, with full hardware accounting given in Supplementary Section~S8. The $32$-qubit two-copy point in Fig.~\ref{fig:measure} corresponds to $\nq = 16$ at $q = 8$ used as a read-out scaling demonstration, and is distinct from the $k = 2$ to $3$, $q = 5$ ($20$ to $30$ qubit) range relevant for the ERA5 case study.

Predictive utility in ERA5 is not carried by one statistic and is not exhausted by covariance: it is distributed over a small number of complementary low-order sectors, and enlarging the stored information beyond them, either to the complete pairwise joint here or to the high-order closures of Case~1, changes the rollout metrics little. The two cases therefore separate what the memory holds from what the predictor consumes.

\begin{figure}[htbp]
\centering
\includegraphics[width=0.95\textwidth]{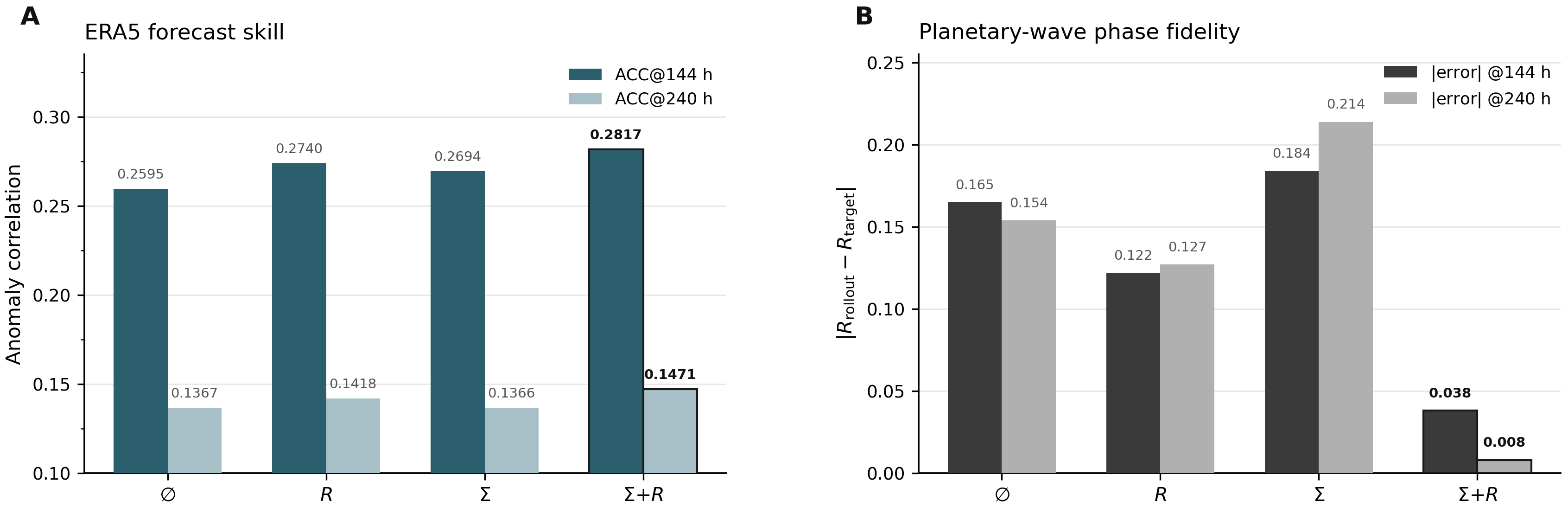}
\caption{Low order invariant constraints in ERA5 $Z_{500}$ forecasting. (A)~Anomaly correlation at 144 and 240\,h for matched continuation runs with no structural prior ($\varnothing$), the inter band wave 5 phase coherence constraint ($R$), the diagonal covariance constraint ($\Sigma$), and their combination ($\Sigma + R$). The phase constraint improves forecast skill when used alone, while the combined constraint gives the highest long lead ACC among the four continuation arms. (B)~Corresponding deviation of the rolled-out phase coherence from its invariant target, $|R_{\mathrm{rollout}} - R_{\mathrm{target}}|$. Within this diagnostic, the covariance constraint alone does not reproduce the selected phase target, whereas $\Sigma + R$ closely restores it at both lead times. The two panels indicate that diagonal covariance and non-diagonal phase coherence provide complementary low-order invariant information. This matched continuation diagnostic is distinct from the full-budget covariance-QIML benchmark. Protocol, $\lambda_R$ sensitivity and the full ACC/RMSE table are given in Supplementary Section~S6.}
\label{fig:phase}
\end{figure}

\section{Discussion}
\label{sec:discussion}

This work connects a quantum-information resource to a physical question about chaotic dynamics. The quantum statistical memory combines compact representation with collective \emph{post hoc} interrogation of a broad candidate structure space. The theory establishes a separation in a task-relevant resource for a general class of states; the hardware runs execute the collective protocol at $\nq = 16$ per copy; and the case studies identify which invariant statistics carry predictive value. In turbulent channel flow, phase correlations persist through eighth order, while most predictive gains arise from low-order constraints. In ERA5, planetary-wave phase coherence provides predictive information complementary to covariance.

The evidence is a chain of four measurements. Representation: the Q-Prior reaches a relative covariance error of $2.31 \times 10^{-4}$ against $3.74 \times 10^{-3}$ for the best compressed classical family at matched parameter budget (Supplementary Section~S4). Access: the collective protocol runs on IQM processors at $\nq = 16$ per copy (Fig.~\ref{fig:measure}), and the complementary diagonal extraction recovers $T_Q^{\mathrm{sim}} = 0.963$ in simulation and $T_Q^{\mathrm{hw}} = 0.66 \pm 0.088$ on hardware. Physical content: the turbulent velocity-direction coherence is a non-diagonal correlator of the invariant measure, and the closure hierarchy above it stays strongly resolved at eighth order while remaining cheap to serve. Downstream value: invariant-statistics regularisation stabilises both rollouts, with a scale-matched classical covariance control identifying it as the stabilising mechanism, and in ERA5 the wave-5 phase sector improves forecast skill without the covariance prior and adds further long-lead skill when combined with it.

Two features of that chain carry beyond the present systems. The extracted statistic is not arbitrary: Q-Priors target the invariant measure of the underlying chaotic system, an object with a long history in dynamical-systems theory and a direct role in long-horizon stability.

The architecture is not tied to the two systems studied here. Candidate applications include systems in which long-time observables are controlled by stationary or invariant distributions, and in which learned predictors lose stability over long rollouts. Once trained, a Q-Prior holds a compact representation of selected invariant statistics in its circuit parameters $\theta$, and the classical model draws on them at inference with no further data sent to the device. The quantum computer works as a domain-specialised coprocessor, supplying a small set of statistics much as specialised accelerators do in classical computing. The advantage also outlasts the present hardware generation: with error correction the same construction runs at higher fidelity and larger $k$, and the classical interface itself becomes worth revisiting, as we discuss below.

What a quantum memory can reach and what the downstream task consumes have different answers. The first is that the invariant measure stays organised far above pairwise order; the second is that the predictor draws on very little of it. The gap between them is the second finding of this work.

The ERA5 phase experiment makes the second answer precise before the saturation appears. Predictive information is not confined to one scalar covariance: the wave 5 phase sector is useful on its own and complementary to $\Sigma$. Yet both are low-order quantities, and once these dominant low-order sectors are supplied, substantially richer information gives diminishing returns. The complete ERA5 pairwise joint, a strictly larger object at the same spatial order, adds only about $1.2\%$ relative ACC at 144\,h beyond the scalar second-order target. The turbulent closures tell the same story from the other end: the structure is physically strong through eighth order and cheap to serve (Section~\ref{sec:case1}), yet carrying it through the prior into the Koopman rollout moves held-out agreement by $+0.2$~pp, well inside the run-to-run spread, with a response that is not monotonic in order.

We call this high-order statistical richness with effective low-order predictive sufficiency: across the two systems and the rollout metrics used here, the invariant measure remains statistically rich at high order while predictive utility concentrates in a small set of complementary low-order sectors, sufficiency referring to that set rather than to covariance alone. The higher observables remain physically informative; it is their marginal contribution to the rollout metrics that saturates once the dominant low-order sectors are in place. A plausible interpretation is that the dominant rollout failure mode is already visible in low-order statistics. In both systems, constraining these statistics removes most of the measurable drift, leaving relatively little additional error for richer closures to correct. The observed saturation is therefore consistent with an effectively low-order failure mode even when the invariant measure itself is not. Whether the underlying dynamics admit a low-order closure in any stricter sense is a separate and open question. The experiments therefore point to a specific target: the minimal invariant-statistical description sufficient to stabilise a learned chaotic system.

These results distinguish the statistical structure represented by the memory from the subset that improves a given predictor. \emph{Post hoc} access allows that subset to be selected after measurement rather than fixed at representation time.

The question we began with becomes a more specific one. High-order physical information demonstrably exists in these invariant measures and is reachable at modest quantum copy cost, so what remains open is when it becomes necessary for macroscopic prediction, and where the sufficiency observed here breaks down. Extreme event statistics, intermittency, rare transitions, non-Gaussian transport and multiscale coupling provide direct tests, because their predictive statistics are expected to depend more strongly on higher-order structure.

One reason the higher orders may saturate lies in the architecture. Collective read-out exposes closure magnitudes, while the downstream experiment reconstructs complex closures from samples of the trained generative prior. Both routes ultimately reduce the statistical memory to a small classical interface before it reaches the predictor. Were the predictor itself running on fault tolerant hardware, the memory could be queried coherently instead of through a classical interface, and the high-order structure that saturates here might act on the dynamics in a way the present pipeline simply cannot express. That experiment needs the downstream model on the same device, which places it in the fault-tolerant era rather than the near term.

On the theory side, the corresponding question is whether the prior's downstream value can be established by proof rather than demonstration, which for classical scientific data would require either a quantum data-generating process or a direct theoretical link between separation-regime invariants and long-horizon stability. Proofs, extended protocols, auxiliary figures and hardware accounting are given in the appendices.

Quantum statistical memory exposes a richer statistical object than present predictors need to consume. Collective access keeps that object open to broad interrogation, while the downstream experiments identify which of the interrogated structures actually matter for prediction: store once, interrogate broadly, select what matters.



\section{Data availability}
The ERA5 reanalysis used in the case study and the read-out experiments is publicly available from the Copernicus Climate Data Store at \url{https://cds.climate.copernicus.eu/datasets/reanalysis-era5-pressure-levels}. The turbulent channel-flow DNS data used in the channel-flow case study are archived at Zenodo, \url{https://doi.org/10.5281/zenodo.16994052}. Full dataset accounting is provided in the Supplementary Information (Supplementary Section~S9).

\section{Code availability}
The main theoretical results of this work do not involve code. The case-study and read-out implementations, including the CUDA-Q prototype, the IQM runner, the ERA5 cache used for paired-sample training, and the archived raw shot counts and job IDs from the hardware runs, are available at \url{https://github.com/UCL-CCS/Weather_QIML.git}.

\bibliography{paper}

\section{Acknowledgements}
The authors thank IQM Quantum Computers for access to the Garnet and Emerald processors and for technical discussions on superconducting hardware and roadmaps, and Zohim Chandani from NVIDIA for assistance with GPU software stacks. Some of the performance results have been obtained on systems in the test environment BEAST (Bavarian Energy Architecture \& Software Testbed) at the Leibniz Supercomputing Centre. P.V.C. is grateful for the financial support of the European Commission CompBioMed Centre of Excellence (Grant Nos. 675451 and 823712) and the UK Engineering and Physical Sciences Research Council through UKCOMES (EP/R029598/1) and SEAVEA (EP/W007711/1). P.V.C. is also grateful for the support of DOE INCITE awards (2025 to 2026) for computational resources at the Oak Ridge and Argonne Leadership Computing Facilities as part of the COMPBIO3 project.

\section{Contributions}
M.W. conceived the study, developed theoretical results and derivations, and carried out the main experiments and case studies. X.X. participated in project discussions, contributed to the data preparation and some of the experiments. M.C. supported the computational runs on the BEAST supercomputer and contributed to some of the experiments. P.V.C. provided guidance throughout the evolution of the project, including software and hardware execution and optimisation processes, and provided access to funding and computing resources. M.W. wrote the manuscript. All authors discussed the results and contributed to the overall manuscript.

\section{Competing interests}
The authors declare no competing interests.

\end{document}